\documentclass[runningheads]{llncs}
\usepackage[T1]{fontenc}
\usepackage{graphicx}
\usepackage{amsmath}
\usepackage{amsfonts}
\usepackage{tikz}
\usepackage{url}
\usepackage{orcidlink}

\makeatletter
\newcommand\blfootnote[1]{%
  \begingroup
    \renewcommand\thefootnote{}%
    \renewcommand\@makefnmark{}%
    \renewcommand\@makefntext[1]{##1}%
    \footnotetext{#1}%
  \endgroup
}
\makeatother

\begin{document}
\title{Ransomware Negotiation: Dynamics and Privacy-Preserving Mechanism Design}
\titlerunning{Ransomware Negotiation}

\author{Haohui Zhang\orcidlink{0009-0000-5596-5668}\inst{1} \and
Sirui Shen\orcidlink{0009-0000-4964-0971}\inst{2} \and
Xinyu Hu\orcidlink{0009-0000-5841-8774}\inst{1} \and
Chenglu Jin\orcidlink{0000-0001-6306-8019}\inst{2}}
\authorrunning{H. Zhang et al.}
\institute{University of Twente, Drienerlolaan 5, Enschede, The Netherlands 
\email{\{h.zhang-5,x.hu\}@utwente.nl}\\
\and
Centrum Wiskunde \& Informatica, Science Park 123, Amsterdam, The Netherlands\\
\email{\{sirui.shen,chenglu.jin\}@cwi.nl}}
\maketitle              %
\begin{abstract}
Ransomware attacks have become a pervasive and costly form of cybercrime, causing tens of millions of dollars in losses as organizations increasingly pay ransoms to mitigate operational disruptions and financial risks. While prior research has largely focused on proactive defenses, the post-infection negotiation dynamics between attackers and victims remains underexplored.
This paper presents a formal analysis of attacker–victim interactions in modern ransomware incidents using a finite-horizon alternating-offers bargaining game model. Our analysis demonstrates how bargaining alters the optimal strategies of both parties.
In practice, incomplete information—attackers lacking knowledge of victims’ data valuations and victims lacking knowledge of attackers’ reservation ransoms—can prolong negotiations and increase victims’ business interruption costs. To address this, we design a Bayesian incentive-compatible mechanism that facilitates rapid agreement on a fair ransom without requiring either party to disclose private valuations.
We further implement this mechanism using secure two-party computation based on garbled circuits, thereby eliminating the need for trusted intermediaries and preserving the privacy of both parties throughout the negotiation. To the best of our knowledge, this is the first automated, privacy-preserving negotiation mechanism grounded in a formal analysis of ransomware negotiation dynamics.

\end{abstract}

\blfootnote{Chenglu Jin is (partially) supported by project CiCS of the research programme Gravitation, which is (partly) financed by the Dutch Research Council (NWO) under the grant 024.006.037. We thank Dr. Jie Zhang for his feedback on a draft of this work.}

\section{Introduction}
Ransomware has become one of the most severe cybersecurity threats, inflicting substantial financial and operational damage on individuals, corporations, and public institutions \cite{connolly2019rise,laszka2017on}. Modern strains go far beyond simple file encryption \cite{hernandez2020an}, combining techniques such as credential theft and large-scale data exfiltration \cite{arctic2025threat} to increase leverage over victims and amplify the pressure to pay. This evolution has made ransomware incidents not only more costly, but also strategically complex, as attackers and victims engage in high-stakes negotiations that determine both the ransom amount and the speed of recovery.

According to the Sophos 2025 Ransomware Report \cite{sophos2024}, 59\% of organizations experienced a ransomware attack in 2024, with 70\% of those incidents resulting in data encryption. The economic impact is stark: the median ransom payment surged fivefold from \$400,000 in 2023 to \$2 million in 2024, with the mean payment rising to \$3.96 million. The average ransom demand also increases to \$2.73 million, nearly \$1 million more than the previous year. 

When facing a ransomware attack, victims typically have two undesirable options: refuse to pay and risk permanent data loss and prolonged recovery, or pay the ransom and remain vulnerable to future attacks, with no guarantee that the data can actually be fully recovered. In practice, many victims choose to pay to minimize downtime and mitigate further financial losses \cite{boticiu2024how}. According to the Sophos report, around 56\% of organizations whose data is encrypted are able to recover their data by paying the ransom in 2024.

Business interruption costs often far exceed the ransom itself. Statista reports that the average downtime following a ransomware attack is 24 days, while the average cost per incident reached \$1.85 million in 2023—a 13\% increase over five years \cite{statista2023}. According to Intermedia \cite{intermedia2016}, downtime-related losses—including significant data recovery costs, reduced customer satisfaction, missed deadlines, lost sales, and traumatized employees—could be even more damaging than the ransom payment. Consequently, in many cases, business interruption costs are the largest source of financial loss following a ransomware incident.

Negotiation, as one of the most important steps before the victim pays the ransom, has been largely overlooked in academic research. Few works explore the negotiation dynamics from a strategic perspective \cite{boticiu2024how,faivre2022negotiation,hofmann2020how,kumamoto2023evaluating,cymru2022analyzing}, and even fewer apply game-theoretic methods \cite{meurs2023double,ryan2022dynamics,vakilinia2021mechanism}. This paper fills in this gap by modeling the interactions as a finite-horizon alternating-offers bargaining game and conducting equilibrium analysis under the complete information scenario. With the goal of reducing business interruption costs for the victim, we further propose a bargaining strategy and its corresponding best response in the incomplete information setting. Inspired by this bargaining strategy profile, we design an automated negotiation mechanism that facilitates efficient agreement by providing the attacker with appropriate financial incentives, thereby accelerating the recovery process for the victim.

To protect sensitive information during this high-stakes process, we implement the negotiation mechanism via a secure two-party computation technique known as garbled circuits~\cite{bellare2012foundations}. This approach eliminates the need for trusted intermediaries while ensuring that neither party reveals their private valuations. Even if the negotiation fails, the attacker learns nothing about the victim’s internal data valuation, preserving strategic advantage and minimizing the risk of further exploitation.

\paragraph{Our Contributions}
In this paper, we first present a comprehensive multistage game model for modern ransomware attacks that explicitly incorporates negotiation dynamics. Our analysis demonstrates how bargaining rounds influence optimal strategies for both attackers and victims, establishing the conditions for subgame perfect Nash equilibrium (SPNE) under a complete information scenario and introducing a strategy profile in the incomplete information setting. Furthermore, we design a novel Bayesian incentive-compatible negotiation mechanism and implement our mechanism in a secure two-party computation protocol. Our privacy-preserving negotiation mechanism allows the victim and attacker to rapidly reach a mutually agreed ransom price while maintaining privacy.
To the best of our knowledge, this paper presents the first automated negotiation mechanism implemented using secure two-party computation.

\medskip
The remainder of this paper is organized as follows. Section~\ref{Sec:Lit} reviews major related works on the game-theoretic analysis of ransomware attacks. Section~\ref{Sec:GT} models the interaction between a ransomware attacker and a victim, and presents related equilibrium analysis. Section~\ref{Sec:Mec} presents the details of our proposed mechanism and the corresponding protocol. Section~\ref{Sec:Imp} discusses the garbled circuit implementation of the proposed mechanism. Section~\ref{Sec:Con} concludes the paper.

\section{Related Work} \label{Sec:Lit}
Ransomware has been well studied as a multistage game to thoroughly model the diverse interactions between the attacker and the victim (or defender) over time. Prior research \cite{caporusso2019game,cartwright2019pay,laszka2017on,li2020ransomware,yin2023deterrence,zhang2022multistage} has modeled various aspects of this adversarial interaction, including attack strategies, law enforcement involvement, backup policies, data-selling threats, and defense mechanisms.
Unlike these works, our analysis deliberately excludes the attacker's targeting stage and the victim's backup stage. Instead, we focus on the post-infection negotiation phase, modeling how a victim’s financial loss accumulates over time and analyzing the strategic bargaining processes. While some studies have assumed the absence of any negotiation or bargaining opportunity \cite{li2020ransomware}, in reality, the victims usually have the opportunity to negotiate and bargain, as in CONTI \cite{cymru2022analyzing}. 

Pratical guidance on ransomware negotiation has emerged alongside limited academic research on the topic. As the ransomware threat has intensified, guidance on how to negotiate with attackers began emerging as early as 2016 \cite{cristal2016how}. Since then, many blog posts have been published on the internet by negotiators and cyber insurance providers for victims \cite{brainstorm,vakulovlevelblue}. However, the academic study of ransomware negotiation remains relatively underexplored. One of the earliest formal explorations is by Hofmann \cite{hofmann2020how}, who outlines strategic negotiation practices based on his experience in cyber threat intelligence. Team Cymru \cite{cymru2022analyzing} analyzes negotiation reports from the CONTI group, highlighting that attackers often assess victims' financial positions using public information and adjust their demands accordingly.
Ryan et al. \cite{ryan2022dynamics} focus on targeted ransomware negotiation and model it as an asymmetric non-cooperative two-player game, offering insights into optimal strategies under imperfect information. Similarly, Meurs et al. \cite{meurs2023double} examine double extortion ransomware through a signaling game framework with double-sided information asymmetry. Their results show that when attackers lack precise knowledge of a victim’s data valuation, their expected payoff decreases, thereby lowering ransom demands and discouraging escalation. Conversely, signaling high data value can lead to higher demands and increased leverage for the attacker.
Finally, Boticiu and Teichmann \cite{boticiu2024how} offer a comprehensive overview of ransomware negotiation procedures, drawing from operational insights to detail typical negotiation phases and highlight the importance of timing, information management, and disaster recovery plan.

Prior to this work, only three published studies have modeled ransomware negotiation as a bargaining game: one by Hernández-Castro et al. \cite{hernandez2017economic}, one by Cartwright et al. \cite{cartwright2019pay} and one by Zhang et al. \cite{zhang2025bargaining}; the last is currently released as a preprint. Hernández-Castro et al. \cite{hernandez2017economic} propose a static game model of ransomware attacks, evaluating attacker and victim payoffs under three pricing strategies: fixed-rate pricing, price discrimination, and bargaining. However, the authors directly refer to results from classic bargaining literature without developing a ransomware-specific bargaining model. Cartwright et al. \cite{cartwright2019pay} develop a dynamic game that incorporates bargaining within a broader framework, accounting for law enforcement intervention. They derive optimal strategies under both complete and incomplete information regarding victims' willingness-to-pay. Their findings suggest that a criminal's bargaining power increases with the threat of irrational aggression and is further enhanced by a credible commitment to decrypt files upon payment. This reputation-building dynamic aligns with our findings. However, the bargaining process itself is not modeled as a standalone mechanism in \cite{cartwright2019pay} but rather embedded within a larger strategic setting. 

The work of Zhang et al.~\cite{zhang2025bargaining} is conducted independently and in parallel with ours. They introduce a dedicated Ransomware Bargaining Game (RGB) analysis framework that systematically examines attacker-victim negotiation in ransomware events. Their framework distinguishes between attacker types based on their attitude toward ransom and analyzes the convergence and equilibrium properties across three negotiation formats: one-round, multi-round, and continuing-round RBGs. 
In contrast to \cite{hernandez2017economic,zhang2025bargaining}, our model does not rely on discounting future payoffs. Instead, we formulate attenuation based on the victim's financial losses accumulating over time. Furthermore, unlike \cite{hernandez2017economic,zhang2025bargaining}, where the number of rounds is fixed or based on the victim’s willingness-to-pay or the attacker’s level of greed, the number of rounds in our framework is determined by the reservation values of both the attacker and the victim. Most importantly, distinct from all three prior works, our analysis places a stronger emphasis on the role of complete and incomplete information regarding these reservation values, which significantly shapes the negotiation outcome.

Before this work, only one paper \cite{vakilinia2021mechanism} attempts to design a mechanism to facilitate negotiation between ransomware attackers and victims with the explicit goal of reducing business interruption costs. Vakilinia et al. \cite{vakilinia2021mechanism} define two ransomware dilemma models, explore mechanism design approaches, and propose smart contract-based solutions to eliminate the need for a trusted third party while enabling atomic ransom-key exchanges.
However, their framework prevents any agreement when the attacker's minimum acceptable price exceeds half of the victim's true valuation. Moreover, their model overlooks a critical dynamic: as time passes, a victim accumulates additional losses that effectively reduce her valuation of the encrypted data. This limitation results in their design not creating sufficient incentives for attackers to provide decryption keys rapidly. 
Additionally, like other ransomware-negotiation studies, their work overlooks the strategic importance of keeping victims’ reservation values confidential during bargaining. This privacy protection is critical to maintaining negotiating leverage in situations where information asymmetries exist.

\section{Game-Theoretic Analysis}
\label{Sec:GT}
In this study, we extend the ransomware multistage game model proposed by Caporusso et al. \cite{caporusso2019game} to include reputation systems and negotiation. We emphasize the importance of reputation systems and the negotiation processes for both attackers and victims.
We model the negotiation process as a finite-horizon alternating-offers bargaining game between an attacker and a victim, provide the SPNE based on backward induction under the assumption of complete information, and propose a strategy profile to urge the attacker to quickly reach an agreement with the victim under the assumption of incomplete information. 

\begin{table}[t]
\centering
\footnotesize %
\centering
\caption{Notation Description}
\vspace{-.69mm}
\resizebox{\columnwidth}{!}{
\begin{tabular}{p{1.5cm} p{9cm}}
\hline
Notation &  Description \\
\hline
$v \in \mathbb{R}_{>0}$ & Actual value of the data owned by the victim \\
$c \in \mathbb{R}_{\geq0}$ & Cost of performing the attack and handling the data, with $c_r$ representing the cost for releasing the files and $c_d$ for deleting the files, $c_r > c_d \cong 0$ \\
$r \in \mathbb{R}_{\geq0}$ & Ransom demand proposed by the attacker, with $r_f$ the final ransom decided through negotiation, $r_{min}$ the minimum ransom the attacker can accept ($r_{min} \gg c$), and $r_{max}$ the maximum ransom the victim can pay \\
$\tau \in \mathbb{R}_{\geq0}$ & Trust level or reputation value of the attacker, with $\tau_g$ representing gain in trust and $\tau_l$ representing loss in trust, $\tau_l \cong \tau_g$ \\
$\kappa \in \mathbb{R}_{>0}$ & Credibility of threat posed by the attacker, with $\kappa_g$ representing gain in credibility and $\kappa_l$ representing loss in credibility, $\kappa_l \cong \kappa_g$ \\
\hline
\end{tabular} \label{tab:notation}}
\label{tab:notation}
\end{table}

\subsection{Multistage Game Model}
The ransomware scenario can be modeled as a sequential, multistage game involving interactions between the attacker and victim. 
In this study, we identified three critical and fundamental stages in the process of a ransomware attack. The related notations are defined in Table~\ref{tab:notation}.

(1) \textit{Stage 1 - ransom requested.} The attacker infects the victim's computer systems and issues a ransom demand $r > 0$, typically delivered via email or a designated website.

(2) \textit{Stage 2 - negotiation.} Having seen the demand $r$, the victim takes an action from \textbf{Pay (V1)}, \textbf{Don't pay (V2)}, and \textbf{Make a counteroffer (V3)}. Notably, almost all (genuine) ransomware strains enable some form of communication between the attacker and victims, allowing victims to make a counteroffer \cite{f-secure}. One key reason is that the attacker does not know exactly the actual value of the data $v$ and the highest ransom $r_{max}$ the victim can afford. 
Although the "irrational aggression" (attacker does not accept any counteroffer) can increase the credibility of threats $\kappa$ posed by the attacker and may get a higher optimal ransom demand \cite{cartwright2019pay}, it is not the attacker's interest to make a ransom demand that is not affordable by the victim. 
Therefore, bargaining is a key aspect of the ransomware game \cite{caporusso2019game,cartwright2019pay}. 

Suppose the victim chooses \textbf{V3}, then the attacker will evaluate the new offer and choose to \textbf{Accept (A1)}, \textbf{Don't accept (A2)}, or \textbf{Make a counteroffer (A3)}. If the attacker chooses \textbf{A2}, the victim needs to reconsider the attacker's last offer. Next, it is the victim's turn to react. This bargaining process continues until the victim chooses \textbf{V1} or \textbf{V2}, or the attacker does not have any patience to negotiate.
Usually, if the attacker has no more patience, she will notify the victim and give the victim a last chance to pay. We denote the final amount paid by the victim as $r_f$, normally $r_f \leq r$. 

(3) \textit{Stage 3 - cooperation or defection.} If the victim pays the ransom, the attacker needs to choose between \textbf{Cooperate (A4)} (release the data with cost $c_r$) and \textbf{Defect (A5)} (delete the data with cost $c_d$). Note that random destruction can be equated with \textbf{A5}. If the victim does not choose to compromise, then the attacker needs to choose between \textbf{Release (A6)} and \textbf{Punish (A7)} (delete).

We define a ransomware game that skips negotiation stage and instead takes the final $r_f$ as given. In this formulation, we introduce two reputation parameters for the attacker, the trust level $\tau$ and the credibility of the threat $\kappa$. These capture the effectiveness of any real‐world reputation system, whether maintained by an insurer \cite{stone2020fbi}, an online forum or other intermediary mechanisms \cite{cymru2022analyzing}. In the extreme case of a fully anonymous attacker, we have $\tau=0$ and $\kappa\gg0$. If the reputation system is imperfect, $\tau \approx 0$. The extensive‐form representation of the game, along with the resulting payoff pairs, is shown in Figure~\ref{fig:backward-induction}. We can observe that the best possible outcome for the victims is to receive a $0$ payoff. Since the attacker always acts in response to the victim’s decisions, the victim is placed in a highly disadvantageous position with virtually no bargaining power \cite{hernandez2017economic}.

\begin{figure}[!t]
\tikzset{
  solid node/.style={circle,draw,inner sep=1.2,fill=black},
  hollow node/.style={circle,draw,inner sep=1.2},
}
\centering
\begin{tikzpicture}[font=\footnotesize]
  \tikzset{
    level 1/.style={level distance=15mm,sibling distance=60mm},
    level 2/.style={level distance=15mm,sibling distance=30mm},
  }
 
    \centering
    \node[hollow node,label=above:{Victim}]{V}
    child{node[hollow node,label=left:{Attacker}]{A}
      child{node(l1)[solid node, label=below:{$(-r_f, r_f-c_r+\tau_g)$}]{}
        edge from parent node[left]{$\textbf{A4}$}
      }
      child{node(l2)[solid node, label=below:{$(-r_f-v, r_f - c_d- \tau_l)$}]{}
        edge from parent node[right]{$\textbf{A5}$}
      }
      edge from parent node[left,xshift=-10]{$\textbf{V1}$}
    }
    child{node[hollow node,label=right:{Attacker}]{A}
      child{node(r1)[solid node, label=below:{$(0, -c_r-\kappa_l+\tau_g)$}]{}
        edge from parent node[left]{$\textbf{A6}$}[dashed]
      }
      child{node(r2)[solid node, label=below:{$(-v, -c_d+\kappa_g)$}]{}
        edge from parent node[right]{$\textbf{A7}$}
      }
      edge from parent node[right,xshift=10]{$\textbf{V2}$}
    }
  ;
\end{tikzpicture}
  \caption{Extensive Form of Ransomware Game with Perfect Reputation (The brackets at leaf node are (victim's payoff, attacker's payoff))}
  \label{fig:backward-induction}
\end{figure}
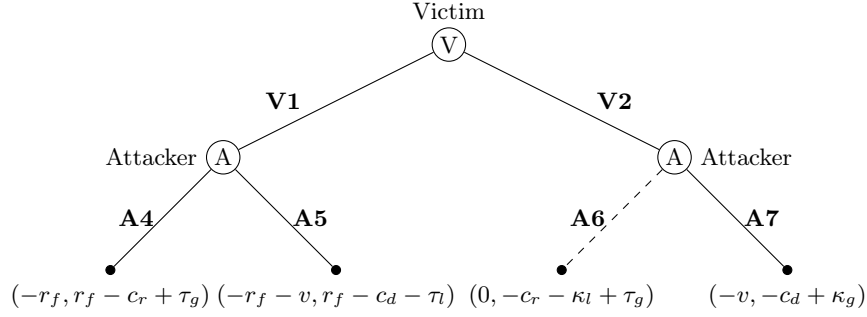

\subsection{Subgame Perfect Nash Equilibrium}
\begin{proposition}
        If the ransomware game is a sequential game where $\tau \approx 0$ and $\kappa > 0$, there exists a unique subgame perfect Nash equilibrium $(\mathbf{V2}, \mathbf{A7})$. 
    \label{thm: ransomware 1.0}
\end{proposition}

\begin{proof}
    Given that there is not a valid reputation system, thus, $\tau_g, \tau_l \approx 0$ and $\kappa_l \cong \kappa_g > 0$. As $c_r > c_d$, through backward induction on the external-form game given in Figure \ref{fig:backward-induction}, $(\textbf{V2}, \textbf{A7})$ is the unique SPNE. \qed
\end{proof}

From Proposition \ref{thm: ransomware 1.0}, we can deduce that victims will only pay the ransom if they believe it offers a good chance of recovering their files, making it essential for attackers to establish a reputation system to build trust with rational victims. 
 
An efficient reputation system with perfect accuracy can be implemented on a blockchain or maintained by authorities or reputable third parties such as insurance companies. As the victim's willingness to pay the ransom highly depends on the attacker's reputation, the increase and decrease of trust and credibility of the threat can directly affect the attacker's future revenue. A natural hypothesis posits that $\tau < \kappa$ and $\tau_l \cong \tau_g > c_r$ as in \cite{caporusso2019game}.
\begin{proposition}
    If the ransomware game is a sequential game where $\kappa > \tau > 0$ and  $\tau_g + \tau_l > c_r$, if $v < r_f$ or $r_{max} < r_f$, then there exists a subgame perfect Nash equilibrium $(\mathbf{V2, A7})$; if $r_f < v$ and $r_f < r_{max}$, then there exists a subgame perfect Nash equilibrium $(\mathbf{V1, A4})$.
    \label{thm:ransome 1.0 reputation}
\end{proposition}
\begin{proof}
    Given that $\kappa > \tau > 0$, we have $\kappa_l+\kappa_g > \tau_g - c_r$. Thus, the attacker's best response to $\textbf{V2}$ is $\textbf{A7}$. Since $\tau_l \cong \tau_g > c_r > c_d$, it follows that $ - c_r + \tau_g > - c_d - \tau_l$, thus, the best response of the attacker to $\textbf{V1}$ is $\textbf{A4}$. If $v < r_f$ or $r_{max} < r_f$, the victim is unable or unwilling to pay the ransom. Through backward induction, we conclude that $(\textbf{V2, A7})$ is the SPNE. If $v > r_f$ and $r_{max} > r_f$, the victim is able and willing to pay the ransom, we conclude that $(\textbf{V1, A4})$ is the SPNE. \qed
\end{proof}
    
Note that if the attacker doesn't allow any counteroffer in stage 2, then $r_f$ can only be $r$ or $0$.
From Proposition \ref{thm:ransome 1.0 reputation}, we obtain two salient points: (1) if the attacker wants to get the ransom from a rational victim successfully, it is necessary for them to be open to counteroffers to make sure that $r_f \leq \min(v, r_{max})$; and (2) the maximum ransom an attacker can get is $\min(v, r_{max})$. 

\subsection{Negotiation Dynamics}
In the above formulation, we ignore the immediate financial loss due to downtime
and the rate at which financial losses accumulate over time. 
To this end, we define the total financial loss at any time $t$, denoted as
\begin{equation}
    L(t) = L_0 + \int_{0}^{t} \ell(\delta)\, d \delta + \mathbb{I}(r_f > 0) \cdot r_f + \xi(t) ,
\end{equation}
where $L_0 \in \mathbb{R}_{\geq 0}$ is defined as the immediate financial loss due to downtime. $t \in \mathbb{R}_{>0}$ is defined as the time elapsed since the ransomware attack began. $\ell(\delta)$ is defined as the loss rate at which financial losses accumulate over time with $\int_{0}^{\infty} \ell(\delta)\, d \delta = v$ and $\ell(\delta) \geq 0$ for all $\delta > 0$. $\mathbb{I}(\cdot)$ is a binary indicator function which $\mathbb{I}(r_f > 0)$ indicates whether the ransom is paid. $\xi(t)$ is a nonnegative, non-decreasing stochastic process, modeling unreimbursable or incalculable losses (e.g. reputation losses, legal costs). The non-decreasing and nonnegativity are defined to reflect the irreversible and potentially escalating nature of these intangible damages over time.

We assume that one bargaining round takes a fixed time $T$; in each round, one party—either the attacker or the victim—proposes a ransom, and the opposing party either accepts the offer or rejects it by proposing a counteroffer, thereby initiating a new round. Consequently, one complete back‐and‐forth exchange requires two rounds, and we denote the total number of rounds by $N$. In the simplest case, if the victim accepts the attacker's initial demand, the negotiation concludes in round $1$. We assume that there is an efficient reputation system and that the attacker will always decrypt the data immediately if the victim pays the ransom. Under these assumptions, we can rationally model the dynamic evolution of the actual value of the encrypted data over time, which is
\begin{equation}
    v(n) = \int_{0}^{\infty} \ell(\delta)\, d \delta - \int_{0}^{n T} \ell(\delta)\, d \delta.
    \label{eq:v_n}
\end{equation} 
We can observe that the actual value of $v$ will gradually decrease with the increase of bargaining rounds $n$. 
The maximum ransom a victim can pay after $n$ rounds of bargaining is $\min\{r_{max}, v(n)\}$, which is the victim's reservation value. $r_{min}$ denotes the attackers' minimum acceptable ransom price where $r_{min} \gg c_r$.  

\begin{figure}[!t]
    \centering
    \includegraphics[width=1\linewidth]{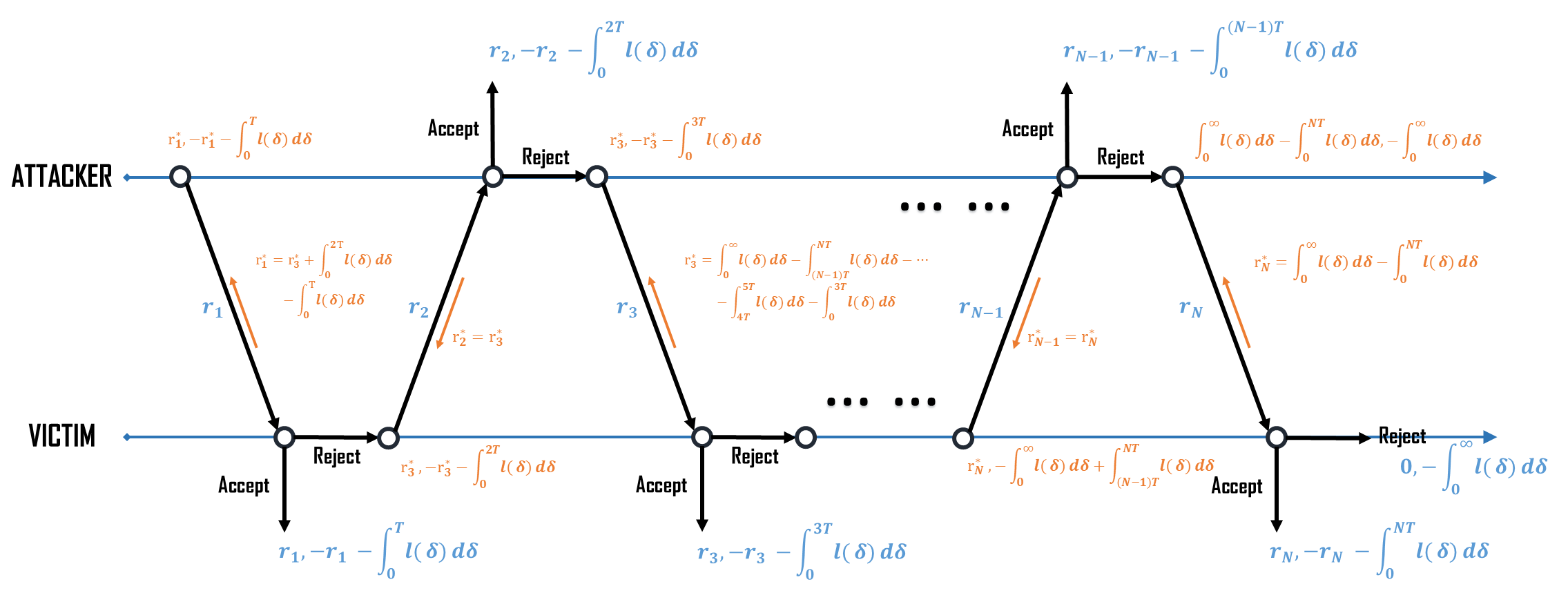}
    \caption{$N$ Rounds Alternating-offers Bargaining Game and Backward Induction}
    \label{fig:bargaining}
\end{figure}

\begin{remark}
    Given that $\xi(t)$ is non-decreasing and the reputation system is efficient, the victim’s optimal strategy is to conclude the negotiation as quickly as possible at the lowest acceptable ransom $r_f$.
\end{remark}

The negotiation dynamics can be modeled as an $N$ rounds alternating-offers bargaining game as shown in Figure \ref{fig:bargaining}. The victim's payoff function is $- L(t_f)$, where $t_f$ denotes the time from the start of the ransomware attack to the time the attacker decrypts the data. The attacker's payoff is $r_f$. Any negotiated agreement within the range $\min\{r_{max},v(n)\}$ and $r_{min}$ at each round $n$ is preferable to no agreement at the end.
If no deal is reached, the attacker’s payoff is $0$ and the victim suffers an accumulated loss of $L(\infty)$. During bargaining analysis, we treat the ransom price and the victim's loss accumulation rate as the only variable parameters. Both agents aim to maximize their payoffs. 
Offers proposed in odd rounds are proposed by the attacker, and offers proposed in even rounds are proposed by the victim. Moreover, $r_{2k} < r_{2k-1}$ and $r_{2k} < r_{2k+1}$ for all $k \in \{1, \dots, \lceil \frac{N-1}{2} \rceil\}$, as otherwise, one party will directly accept the offer.

\subsection*{Complete Information} 
We first assume that both the reservation values of the attacker and the victim are complete information. This assumption, while a simplification of reality, can be justified as follows: $r_{min}$ can be inferred from historical incidents involving the same ransomware group, and $v$ and $\ell(\delta)$ can be estimated based on publicly available financial statements or confidential financial information stolen by the attacker.
We also assume that the marginal loss incurred in the bargaining round is negligible compared to the cumulative future losses for all $n$, i.e. $\int_{(N-1)T}^{NT} \ell(\delta)\, d \delta \ll  \int_{(N+1)T}^{\infty} \ell(\delta)\, d \delta$. Suppose $N$ is the last bargaining round between the attacker and the victim. Then, $r_{min} < v(N)$ and $r_{min} > v(N+1)$, $N \in \mathbb{N}_{\geq 0}$. Without loss of generality, we assume $r_{max} \geq v(1)$.

\begin{proposition}
    If the reservation values are complete information, the subgame perfect equilibrium (the optimal offer given by the attacker and the victim in each round $n$ where $n \in \{1, \dots, N\}$ and $N$ is an odd positive integer) is \begin{equation}
        r^*_n = \int_{0}^{\infty} \ell(\delta)\, d \delta - \sum_{k = 0}^{\lfloor\frac{N-n}{2} -1 \rfloor} \int_{(N-1-2k)T}^{(N-2k)T} \ell(\delta)\, d \delta -  \int_{0}^{(2 \lfloor \frac{n}{2} \rfloor +1) T} \ell(\delta)\, d \delta .
        \label{eq:opt_bargain}
    \end{equation}
    \label{pro:opt_bargain}
\end{proposition}
\begin{proof}
    Given that $r^*_N = \int_{0}^{\infty} \ell(\delta)\, d \delta - \int_{0}^{N T} \ell(\delta)\, d \delta$, based on the backward induction in Figure~\ref{fig:bargaining} (orange part), we obtain $r^*_{N-1} = r^*_N$; $r^*_{N-2} = \int_{0}^{\infty} \ell(\delta)\, d \delta - \int_{(N-1)T}^{N T} \ell(\delta)\, d \delta -  \int_{0}^{(N-2)T} \ell(\delta)\, d \delta$; ...; 
    $r^*_2 = r^*_3$;
    $r^*_1 =  \int_{0}^{\infty} \ell(\delta)\, d \delta - \int_{(N-1)T}^{N T} \ell(\delta)\, d \delta - \dots - \int_{2T}^{3 T} \ell(\delta)\, d \delta - \int_{0}^{T} \ell(\delta)\, d \delta$, which can be encoded into the closed‐form in (\ref{eq:opt_bargain}). 
     
    If $N$ is an even positive integer, as through similar backward induction, it follows that $r^*_N = 0$ and $r^*_{N-1} = \int_{(N-1)T}^{NT} \ell(\delta)\, d \delta$. Given that $r_{min} > v(N+1) = \int_{(N+1)T}^{\infty} \ell(\delta)\, d \delta \gg \int_{(N-1)T}^{NT} \ell(\delta)\, d \delta$ based on the definition of $\ell(\delta)$, we obtain that $r^*_N \leq r^*_{N-1} < r_{min}$. Because the attacker will never propose a ransom less than $r_{min}$, this creates a contradiction. 
    \qed
\end{proof}

In practice, the victim can always propose a counteroffer back in round $N+1$, but because $r_{N+1} < r_N$ and $r_{min} > v(N+1)$, the attacker will always reject and end the negotiation.
As the victim will never accept if the demand is greater than its reservation value, we define the optimal offer function $R: \{(n,N)| 1 \leq n \leq N \} \longrightarrow \mathbb{R}$, mapping each pair $(n,N)$ to a real‐valued offer: \begin{equation}
    R(n,N) = r^*_{n|n\in\{1,\dots,N\}} .
\end{equation}
\begin{lemma}
    The optimal offer function $R$ is monotonically non-increasing. That is $R(k+1,N) \leq R(k,N)$ for all $k \in \{1, \dots, N-1\}$ and $R(n,N+2) \leq R(n,N)$ for all odd integer $N$ and all $n \in \{1, \dots, N\}$.
    \label{lem:non-increasing}
\end{lemma}

\begin{proof}
    Suppose that, contrary to the Lemma, there exists a $R(k+1,N) > R(k,N)$. Based on the definition of $R$ and (\ref{eq:opt_bargain}), we know $R(k+1,N) - R(k,N) = r^*_{k+1}-r^*_{k},$ if $k$ is an even number, $r^*_{k+1}-r^*_{k} = 0$, contradicts, if $k$ is odd,
    \begin{equation*}
       r^*_{k+1}-r^*_{k} = - \int_{0}^{(k+2)T} \ell(\delta)\, d \delta + \int_{0}^{kT} \ell(\delta)\, d \delta = - \int_{kT}^{(k+2)T} \ell(\delta)\, d \delta > 0,
    \end{equation*}
    contradicts the definition of $\ell(\delta)$.
    
    Given (\ref{eq:opt_bargain}), as 
    \[\sum_{k = 0}^{\lfloor\frac{N-n}{2} -1 \rfloor} \int_{(N-1-2k)T}^{(N-2k)T} \ell(\delta)\, d \delta  \leq \sum_{k = 0}^{\lfloor\frac{N+2-n}{2} -1 \rfloor} \int_{(N+1-2k)T}^{(N+2-2k)T} \ell(\delta)\, d \delta , \]
    then, $ R(n,N+2) \leq R(n, N+1) = R(n,N)$.
    \qed
\end{proof}

\begin{lemma}
   $R(n,N) \leq v(n)$ for all odd integer $N$ and $n \in \{1, \dots, N\}$. 
   \label{lem:boundness}
\end{lemma}

\begin{proof}
    Given (\ref{eq:opt_bargain}), as $(2 \lfloor \frac{n}{2} \rfloor +1) T \geq nT$ and $\ell(\delta) \geq 0$, then, 
    \[ \sum_{k = 0}^{\lfloor\frac{N-n}{2} -1 \rfloor} \int_{(N-1-2k)T}^{(N-2k)T} \ell(\delta)\, d \delta +  \int_{0}^{(2 \lfloor \frac{n}{2} \rfloor +1) T} \ell(\delta)\, d \delta \geq \int_{0}^{nT} \ell(\delta)\, d \delta .\]
    \qed
\end{proof}

\begin{remark}
    Given that the reservation values are complete information and $0 < r_{min} \leq r_{max}$, the attacker can get a ransom in round $n$ if and only if $r_{min} \leq R(n,N)$. The possible maximum ransom the attacker can get is $R(1,N)$.
    \label{lem:attacker_min}
\end{remark}

\vspace{-1em}
Therefore, if the attacker wants to get any ransom eventually, either the attacker needs to accept an offer proposed by the victim before round $N$, where $N$ is an odd integer ($r_{min} \leq R(N,N)$ and $r_{min} > R(N+2, N+2)$), or the attacker proposes an offer no larger than $R(2k+1,N)$ in round $2k+1$ where $k \in \{0, \dots, \lfloor \frac{N-1}{2} \rfloor\}$, since otherwise, the victim will never accept it. The optimal strategy profile is for the attacker to propose $R(1,N)$ in round 1, and for the victim to accept it in round 2.

\paragraph{Special Cases}
In some ransomware settings, there is no loss accumulation rate over time but rather a fixed loss if the data or files are never decrypted; in that case, the total cost remains constant unless a ransom $r_f$ is paid. We define
\[L = L_0 + \mathbb{I}(a^A \notin \{\textbf{A4},\textbf{A6}\}) \cdot v + \mathbb{I}(r_f > 0) \cdot r_f,\]
where $a^A$ is the final action of the attacker. Suppose the attacker's reservation value is $r_{min}$ where $r_{min} \gg c$. If $r_{min} > \min \{v, r_{max}\}$, the victim will never pay the ransom. 
When $r_{min} \leq \min \{v, r_{max}\}$ and both sides' reservation values are complete information, based on the unique equilibrium of the infinite‐horizon alternating‐offers bargaining game in \cite{rubinstein1982perfect} (the discount factors $\gamma^A = \gamma^V \rightarrow 1$), the attacker and the victim agree to set the ransom:
\[
r_f     \;=\; 
\frac{\min\{v,\,r_{\max}\} + r_{\min}}{\,2\,}.
\]

\subsection*{Incomplete Information - Private Reservation Values}     
Suppose $r_{min}$ is sufficiently small, which means the attacker can bargain with the victim forever, then $N \rightarrow \infty$. 
Suppose the loss rate $\ell(\delta)$ varies much more slowly than the window‑length $T$, which means for all $n \in \mathbb{N}_{\geq 0}$ each block integral $\int_{nT}^{(n+1)T} \ell(\delta) \, d\delta$ changes negligibly from one block to the next, it follows that
\begin{align*}
    r_1^* &= \lim_{N \rightarrow \infty} \int_{0}^{\infty} \ell(\delta)\, d \delta - \sum_{k = 0}^{\lfloor\frac{N-3}{2}\rfloor} \int_{(N-1-2k)T}^{(N-2k)T} \ell(\delta)\, d \delta -  \int_{0}^{T} \ell(\delta)\, d \delta \\ &= \sum_{m=0}^{\infty} \int_{(2m+1)T}^{(2m+2)T} \ell(\delta)\, d \delta \approx \frac{1}{2}\int_{0}^{\infty} \ell(\delta)\, d \delta.
\end{align*}
From Lemma~\ref{lem:non-increasing}, $R(1,\infty) = \min_{N\geq 1}R(1,N)$, the minimum ransom the victim pays at round $1$ is $\min \big\{ \frac{1}{2}\int_{0}^{\infty} \ell(\delta)\, d \delta, \, r_{max} \big\} $. Consequently, if the attacker’s opening demand falls below this threshold, a good strategy for the victim is to accept it immediately to avoid additional bargaining rounds. Given that the attacker may salt prices to secure a better offer, when designing the victim's bargaining strategy, the victim needs to encourage the attacker to compromise as quickly as possible while ensuring that any attacker's deviation behavior from truthful play cannot improve her expected payoff. We define the truthful play of the attacker as \textit{the attacker accepting the offer $r_n$ in the round $n$ whenever $r_{min} \leq r_n$.}
 
\begin{proposition}
    Suppose that the reservation values are incomplete information.
    Without loss of generality, we assume that $r_{max} \geq v(1)$. We denote $\tilde{r}_2 = q \cdot\int_{0}^{\infty} \ell(\delta)\, d \delta - \int_{0}^{2T} \ell(\delta)\, d \delta$.
    Given the strategy of the victim in round $2$ and $4$ as
    \begin{align*}
        a^V_2 &= \begin{cases}
            \textbf{V1} \quad \quad \quad \quad \quad \quad \quad \quad \quad \quad \quad \ &\text{if } r_1 \leq \frac{1}{2}\int_{0}^{\infty} \ell(\delta)\, d \delta, \\
            \textbf{V3} \ \& \ r_2 = \tilde{r}_2   &\text{otherwise with prob } \overline{p}, \\
            \textbf{V3} \ \& \ r_2 = \int_{2T}^{\infty} \ell(\delta)\, d \delta &\text{otherwise with prob } 1-\overline{p} ,
        \end{cases} \\
        a^V_{4} &= \begin{cases}
            \textbf{V1} \quad \quad \quad \quad \quad \quad \quad \quad \quad \quad \quad  \ &\text{if } r_{3} \leq \int_{3T}^{\infty} \ell(\delta)\, d \delta \text{ and with prob }\rho,\\
            \textbf{V2} &\text{otherwise}, \\
        \end{cases} 
    \end{align*}
    where $\frac{\int_{0}^{2T} \ell(\delta)\, d \delta}{\int_{0}^{\infty} \ell(\delta)\, d \delta} \leq q \leq \frac{\int_{0}^{2T} \ell(\delta)\, d \delta}{\int_{0}^{3T} \ell(\delta)\, d \delta}$, 
    $ \max \big\{0, \frac{q \cdot \int_{0}^{\infty}\ell(\delta)\, d \delta - \int_{0}^{2T}\ell(\delta)\, d \delta}{\int_{3T}^{\infty}\ell(\delta)\, d \delta} \big\} \leq \rho \leq 1 $, and \\
    $\frac{\int_{2T}^{\infty} \ell(\delta)\, d \delta}{\rho\cdot \int_{3T}^{\infty} \ell(\delta)\, d \delta + (1-q) \cdot \int_{0}^{\infty} \ell(\delta)\, d \delta} \leq \overline{p} \leq 1$,
    the best response of attacker in round $3$ is 
    \begin{align*}
        a^A_{3} &= \begin{cases}
            \textbf{A1}  \quad \quad \quad \quad \quad \quad \quad \quad \quad \quad \quad \ &\text{if } r_{min} \leq r_{2} , \quad \quad \quad \quad \quad \quad \quad \quad \quad \ \\
            \textbf{A3} \ \& \ r_{3} = \max\{\frac{r_2}{q}, r_{min}\} &\text{otherwise. } \\
        \end{cases} 
    \end{align*}
    \label{pro:opt_bargian_incomplete}
\end{proposition}

\begin{proof}
When $r_2 = \tilde{r}_2$ and $r_{min} \leq \tilde{r}_2$, the expected payoff of the attacker to act truthfully in round $3$ is
\begin{align*}
    \mathbb{E}\Big[u^A_3\Big(\textbf{A1} \Big| r_2 = \tilde{r}_2 \geq r_{min} \Big) \Big]
    = (1 - \overline{p} + \overline{p}\,q) \, \int_{2T}^{\infty} \ell(\delta)\, d \delta + \overline{p} \,(q-1) \,\int_{0}^{2T} \ell(\delta)\, d \delta .
\end{align*}
If the attacker chooses $\textbf{A3}$ (to make a counteroffer), the offer $r_3$ given by the attacker should be larger than $r_2$. 
As $r_2$ has $1-\overline{p}$ probability to be $\int_{2T}^{\infty} \ell(\delta)\, d \delta $, if so and the attacker acts untruthfully, $\mathbb{E}[ u^A_3(\textbf{A3}|r_2 = \int_{2T}^{\infty} \ell(\delta)\, d \delta)] = 0$. When $r_2 = \tilde{r}_2$, given that the victim will accept any offer below $\int_{3T}^{\infty} \ell(\delta)\, d \delta$ in round $4$, then the maximum possible ransom the attacker can get is $\int_{3T}^{\infty} \ell(\delta)\, d \delta$. Thus, 
\[ \mathbb{E} \Big[u^A_3\Big(\textbf{A3} \Big| r_2 = \tilde{r}_2 \geq r_{min} \Big) \Big] \leq \overline{p} \cdot \rho \cdot \int_{3T}^{\infty} \ell(\delta)\, d \delta.\]
Given that 
\[ 
 \rho \geq \frac{q \cdot \int_{0}^{\infty}\ell(\delta)\, d \delta - \int_{0}^{2T}\ell(\delta)\, d \delta}{\int_{3T}^{\infty}\ell(\delta)\, d \delta}
 \text{ and } 
 \overline{p} \geq  \frac{\int_{2T}^{\infty} \ell(\delta)\, d \delta}{\rho\cdot \int_{3T}^{\infty} \ell(\delta)\, d \delta + (1-q) \cdot \int_{0}^{\infty} \ell(\delta)\, d \delta} ,
 \]
we obtain \[\mathbb{E} \big[u^A_3 (\textbf{A3} | r_2 = \tilde{r}_2 \geq r_{min}) \big] \leq \mathbb{E} \big[u^A_3 (\textbf{A1} | r_2 = \tilde{r}_2 \geq r_{min}) \big].\] 
The best response of the attacker in round $3$ when $r_2 = \tilde{r}_2 \geq r_{min}$ is $\textbf{A1}$ (to accept). 
Thus, if a rational attacker rejects in round $3$, $r_{min} > \tilde{r}_2$. 

Based on Remark~\ref{lem:attacker_min}, the attacker can only get a ransom when $r_{min} \leq v(2) = \int_{2T}^{\infty} \ell(\delta)\, d \delta$ in round $2$. 
When $r_{min} > v(2)$, the attacker's payoff will always be $0$.

When $\tilde{r}_2 < r_{min} \leq \int_{2T}^{\infty} \ell(\delta)\, d \delta = r_2$, the attacker's best response is $\textbf{A1}$, since $r_2 = v(2)$ (which exceeds $v(3)$) is the highest ransom the attacker can obtain.

When $ r_2 = \tilde{r}_2 < r_{min} \leq \int_{2T}^{\infty} \ell(\delta)\, d \delta$, the best response of the attacker is $\textbf{A3}$. Since $\frac{\int_{0}^{2T} \ell(\delta)\, d \delta}{\int_{0}^{\infty} \ell(\delta)\, d \delta} \leq q \leq \frac{\int_{0}^{2T} \ell(\delta)\, d \delta}{\int_{0}^{3T} \ell(\delta)\, d \delta}$, it follows that $\frac{r_2}{q} \leq v(3)$. Therefore, the maximum offer can be proposed based on the attacker's current information is $\frac{r_2}{q}$ and the optimal counteroffer is $ \max\{ \frac{r_2}{q}, r_{min} \}$. Notably, when $v(3) < r_{min} \leq v(2)$, the attacker will not get any ransom in this case.
\qed
\end{proof}

\section{Privacy-preserving Negotiation Mechanism} \label{Sec:Mec}

In this section, under the assumption of private reservation values (incomplete information), we design a novel Bayesian incentive-compatible mechanism to help the victim and the attacker automatically bargain and quickly reach a consensus, while preserving the privacy of both sides through garbled circuits.

\subsection{Mechanism Design}
\label{SubSec:Mec}
Since reaching an agreement quickly benefits both the attacker and the victim, our mechanism design builds on Proposition~\ref{pro:opt_bargian_incomplete}. By implementing it via garbled circuits that directly compute the final ransom, we eliminate negotiation delays and thus adopt a simplified strategy profile based on the proposition.
We assume that $\xi(t)=0$ for all $t$. We denote the victim's reservation value function based on the bargaining time $t$ as 
\begin{equation}
    \psi(t) = \min \bigg\{ \int_{t}^{\infty}\ell(\delta) \, d \delta , \, r_{max} \bigg\} .
\end{equation}

\begin{definition}[Mechanism Design for Ransomware Negotiation]
Consider a negotiation mechanism $\mathcal{M}$ between a victim and an attacker, where each agent reports a type: the victim reports $\hat{\theta}^V \in \mathbb{R}_{\geq 0}$ and the attacker reports $\hat{\theta}^A \in \mathbb{R}_{\geq 0}$,  representing their reservation values.
    Given the victim's strategy $\pi^V$:
     \begin{align*}
        a^V_2 &= \begin{cases}
            \textbf{V3} \ \& \ r_2 = q \cdot \hat{\theta}^V   &\text{with prob } \overline{p}, \\
            \textbf{V3} \ \& \ r_2 = \hat{\theta}^V \quad \quad &\text{with prob } 1-\overline{p} ,
        \end{cases} \\
        a^V_{4} &= \begin{cases}
            \textbf{V1} \ (r_f=r_3) \quad \ \ &\text{if } r_{3} \leq  \hat{\theta}^V \text{ and with prob } q,\\
            \textbf{V2} \ (r_f=0, \sigma=1)  &\text{if } r_{3} \leq  \hat{\theta}^V \text{ and with prob } 1-q,\\
            \textbf{V2} \ (r_f=0) &\text{otherwise}, \\
        \end{cases} 
    \end{align*}
    where $q, \overline{p} \in [0,1]$ and $\overline{p} \cdot (1-q) = \frac{1}{2}$, and the best response of the attacker $\pi^A$:
    \begin{align*}
        a^A_{3} &= \begin{cases}
            \textbf{A1} \ (r_{f}=r_2) &\text{if } \hat{\theta}^A \leq r_{2}, \\
            \textbf{A3} \ \& \ r_{3} = \max\{\frac{r_2}{q}, \hat{\theta}^A\} &\text{otherwise, } \\
        \end{cases} 
    \end{align*}
    where $r_f$ is the final ransom determined by the strategy profile $\boldsymbol{\pi} = (\pi^V, \pi^A)$, we define the \textbf{allocation rule} as:
    \begin{align*}
        \alpha^V = \alpha^A =  
        \begin{cases}
            1 & \text{if } r_f > 0 \text{ or } \sigma=1, \\
            0 & \text{otherwise},
        \end{cases}
    \end{align*}
    and the corresponding \textbf{payment rule} as:
    \[
        \beta^V = \beta^A = r_f .
    \]  
    Then, the mechanism $\mathcal{M}$ maps input types $(\hat{\theta}^V, \hat{\theta}^A)$ to outcome $(\boldsymbol{\alpha}, \boldsymbol{\beta})$. ($\overline{p}$ and $q$ are exogenously chosen “coin‐flip” biases.)

\end{definition}

The mechanism assumes a quasilinear utility, just as other auction mechanisms. A victim with valuation $\theta^V$ receives a utility $u^V(\theta^V,\hat{\theta}^V) = (\alpha^V\cdot \theta^V - \beta^V)$ for reporting type $\hat{\theta}^V$, while an attacker with valuation $\theta^A$ receives a utility $u^A(\theta^A,\hat{\theta}^A) = (\beta^A-\alpha^A\cdot \theta^A) $ for reporting type $\hat{\theta}^A$. 
Notably, the attacker's valuation $\theta^A$ in this mechanism is $r_{min}$, and the victim's valuation $\theta^V$ is $\psi(t)$ at time $t$. We assume a common prior in which $\theta^A$ and $\theta^V$ are drawn independently and uniformly from a common uniform distribution over $[\underline{r}, \overline{r}]$. Each player knows her own type and holds the prior as her belief about the other’s type.

\begin{theorem}
    Suppose $\xi(t)=0$ for all $t$, and attacker's true valuation $\theta^A$ is drawn independently and uniformly from a common prior uniform distribution $F$ over $[\underline{r}, \overline{r}]$, the mechanism $\mathcal{M}$ is Bayesian incentive-compatible.
\end{theorem}

\begin{proof}
Given the input $\hat{\theta}^V$ of the victim at time $t$ and input $\hat{\theta}^A$ of the attacker to the mechanism $\mathcal{M}$, the expected utility of the attacker is
\begin{align*}
    & \mathbb{E}\big[ u^A\big(\theta^A, \hat{\theta}^A; \hat{\theta}^V \big)\big] \\ &= \begin{cases}  \overline{p} \, (q\hat{\theta}^V-\theta^A) + (1-\overline{p})\, (\hat{\theta}^V-\theta^A) & \text{if } \hat{\theta}^A \leq q\hat{\theta}^V, \\  
    (1-\overline{p})\, (\hat{\theta}^V-\theta^A) + \overline{p} \,  q \, (\hat{\theta}^V-\theta^A) + \overline{p}\,(1- q)\,(-\theta^A) &\text{if }q\hat{\theta}^V<\hat{\theta}^A \leq \hat{\theta}^V , \\
     0 & \text{otherwise}. \\
    \end{cases}
\end{align*}

Fix the attacker’s true type $\theta^A\ge0$ and an arbitrary victim report $\hat\theta^V\ge0$.  We compare the attacker’s utility when she reports $\hat\theta^A$ to that when she reports truthfully, $\hat\theta^A=\theta^A$.

\smallskip\noindent\textbf{Case 1: \(\hat\theta^V<\theta^A\).}  
Then \(q\hat{\theta}^V<\theta^A\) as well, so in every branch \(q\hat{\theta}^V-\theta^A<0\) and \(\hat\theta^V-\theta^A<0\).  Hence, the expected utility when $\hat{\theta}^A\leq \hat{\theta}^V$ is nonpositive, and the expected utility when $ \hat{\theta}^A > \hat{\theta}^V$ gives zero.  Thus, truthful reporting \(\hat\theta^A=\theta^A\) achieves this maximum.

\smallskip\noindent\textbf{Case 2: \(q\hat{\theta}^V<\theta^A\le\hat\theta^V\).}  
Here \(\hat\theta^V-\theta^A\ge0\) but \(q\hat{\theta}^V-\theta^A<0\). 
\begin{align*}
    \hat\theta^A\le q\hat{\theta}^V
&\;\implies\;
\mathbb{E}[u^A] \;=\; \bigl((1-\overline p)+\overline p\,q\bigr)\,\hat\theta^V-\theta^A, \\
q\hat{\theta}^V<\hat\theta^A\le\hat\theta^V
&\;\implies\;
\mathbb{E}[u^A] \;=\;\bigl((1-\overline p)+\overline p\,q\bigr)\,\hat\theta^V-\theta^A.
\end{align*}
 In either event, the payoff ties the maximum value, while any \(\hat\theta^A>\hat\theta^V\) yields 0.  

\smallskip\noindent\textbf{Case 3: \(\theta^A\le q\hat{\theta}^V\).}  
Now \(\hat\theta^V-\theta^A\ge q\hat{\theta}^V-\theta^A\ge0\), same to Case 2, 
the utility equals 
\(\bigl((1-\overline p)+\overline p\,q\bigr)\,\hat\theta^V-\theta^A\) no matter \(\hat{\theta}^A < q\hat{\theta}^V \) or \(q\hat{\theta}^V<\hat\theta^A\le\hat\theta^V\), while any \(\hat\theta^A>\hat\theta^V\) yields zero.  

In all three cases, \(\hat\theta^A=\theta^A\) weakly dominates any other report.  Hence, truth‐telling is a (weakly) dominant strategy for the attacker.

\bigskip
For all $t$, the interim expected utility of the victim averaged over the distribution $F$ of attacker's report is, 
\begin{align*}
    \mathbb{E}_{\hat{\theta}^A\sim F} \Big[\mathbb{E} \big[ u^V\big({\theta}^V,  \hat{\theta}^V  ; \hat{\theta}^A \big)\big] \Big]
    = \mathbb{P}[\hat{\theta}^A \leq q\hat{\theta}^V]\cdot \big(\overline{p}\, (\theta^V - q\hat{\theta}^V) + (1-\overline{p})\,(\theta^V - \hat{\theta}^V)\big) \\ + \mathbb{P}[q\hat{\theta}^V < \hat{\theta}^A \leq \hat{\theta}^V]\cdot \big((1-\overline{p})\,(\theta^V-\hat{\theta}^V) +\overline{p}\,q \, (\theta^V-\hat{\theta}^V) + \overline{p} \,(1-q) \,\theta^V \big) . 
\end{align*}

Denote by $F(\hat{\theta}^V)$ the cumulative distribution function of the random variable $\hat{\theta}^A$, $F(x) = \mathbb{P}(\hat{\theta}^A \leq x)$. Then,
\begin{align*}
   \mathbb{E}_{\hat{\theta}^A\sim F} \Big[\mathbb{E} \big[ u^V\big({\theta}^V,  \hat{\theta}^V  ; \hat{\theta}^A \big)\big] \Big] = F(q\hat{\theta}^V)\,\overline{p}\,(1-q)\,\theta^V + F(\hat{\theta}^V)\,(1-\overline{p}+\overline{p}\,q)\,\theta^V \\ 
     -F(\hat{\theta}^V)\,(1-\overline{p}+\overline{p}\,q)\,\hat{\theta}^V + \big(F(\hat{\theta}^V) - F(q\hat{\theta}^V) \big) \cdot \overline{p} \,(1-q) \theta^V.
\end{align*} 
Since in equilibrium the attacker reports truthfully, $\hat{\theta}^A = \theta^A$. By assumption, $\theta^A$ is independent of $ (\theta^V, q, \overline{p})$ and uniformly drawn from $F$, which is $\text{Unif}[\underline{r}, \overline{r}]$, by the monotone linear change of variables $U=(\theta^A - \underline{r})/(\overline{r}-\underline{r})$, we may, without loss of generality, normalize the support to $[0,1]$. Under this normalization, $F(u)=u$ and $F'(u) = f(u)=1$ for $u\in[0,1]$. Then, the partial derivative is
\begin{align*}
    \frac{\partial \mathbb{E} [ u^V ]}{\partial \hat{\theta}^V} =& f(q\hat{\theta}^V) \, q \,\overline{p}\,(1-q)\,\theta^V + f(\hat{\theta}^V)\,(1-\overline{p}+\overline{p}\,q)\,\theta^V - f(\hat{\theta}^V)\,(1-\overline{p}+\overline{p}\,q)\,\hat{\theta}^V \\ &- F(\hat{\theta}^V)\,(1-\overline{p}+\overline{p}\,q) + \big( f(\hat{\theta}^V) - f (q\hat{\theta}^V)\, q \, \big) \,\overline{p}\,(1-q)\, \theta^V \\[5pt]
    = & \overline{p}\,(1-q)\, \theta^V + (1-\overline{p}+\overline{p}\,q) \, \theta^V  - 2 \,(1-\overline{p}+\overline{p}\,q) \hat{\theta}^V .
\end{align*}
Setting the derivative to zero yields the unique interior maximizer
\[ 
    \hat{\theta}^V  = \frac{1}{2\cdot(1-\overline{p}+\overline{p}\,q)}\,{\theta}^V .
\]

Using the condition $\overline{p}\cdot(1-q)=\frac{1}{2}$, we obtain $\hat\theta^V=\theta^V$.
As the second derivative of the expected utility is negative, truthfully reporting \(\hat\theta^V=\theta^V\) maximizes the victim’s (Bayesian) expected utility, which constitutes a Bayes–Nash equilibrium.  
We can conclude that the negotiation mechanism $\mathcal{M}$ holds BIC.
\qed
\end{proof}

In our ransomware‐negotiation mechanism, the only condition needed for Bayesian incentive compatibility is $\overline{p} \cdot (1-q) = \frac{1}{2}$, where $q,\overline{p} \in [0,1]$. If the victim reports their type truthfully,  then whenever the attacker’s report satisfies $\hat{\theta}^A\leq \hat{\theta}^V$, her expected payment is 
$$(1-\overline{p}+\overline{p}\,q)\cdot \hat{\theta}^V = (1-\overline{p}+\overline{p}\,q)\cdot\psi(t) = \frac{1}{2}\psi(t).$$ 
Solving \(\overline{p}(1-q)=\tfrac12\) for \(q\) gives 
\[
q \;=\;1 - \frac{1}{2\overline{p}}, 
\quad \overline{p}\in\bigl[\tfrac12,1\bigr],
\]
or equivalently for \(\overline{p}\) in terms of \(q\), 
\[
\overline{p} \;=\;\frac{1}{2(1-q)},
\quad q\in\bigl[0,\tfrac12\bigr].
\]

Compared with the double-sided-blind auction proposed in \cite{vakilinia2021mechanism}, which only guarantees that the victim pays $\frac{1}{2} \psi(t)$ when $r_{min} \leq \frac{1}{2}\psi(t)$, our negotiation mechanism ensures an expected payment of $\frac{1}{2} \psi(t)$ both when $r_{min} \leq \frac{1}{2}\psi(t)$ and when $\frac{1}{2} \psi(t) < r_{min} < \psi(t)$. 

\subsection{Privacy-Preserving Negotiation Protocol}
We implement mechanism $\mathcal{M}$ using a maliciously secure two-party computation protocol, enabling the victim and the attacker to reach an agreement without revealing their private information. The design of our protocol ensures fairness between the two parties and efficiency of the computation without requiring trust in the other party or third parties.

Our protocol utilizes garbled circuits as its backend technology~\cite{bellare2012foundations}. A garbled circuit scheme involves two parties: a garbler and an evaluator. When running the garbling scheme, the garbler first takes a Boolean circuit $f$, representing the function to be evaluated jointly, produces a garbled circuit $\mathcal{C}$, and sends $\mathcal{C}$ to the evaluator. The garbler also generates labels representing the encrypted form of all possible values on all wires in $\mathcal{C}$. Hereafter, $[\cdot]$ denotes the garbled/encrypted variables. After the circuit and labels are generated, the garbler encodes its input $I_g$ to its garbled/encrypted label $[I_g]$, while the evaluator obtains the labels $[I_e]$ corresponding to its input $I_e$ by running Oblivious Transfer (OT) with the garbler.  According to the security properties of garbled circuits, given input labels $[I]$, the evaluator cannot infer or manipulate the value of $I$ in the function evaluation~\cite{bellare2012foundations}. Then the evaluator computes $\mathcal{C}$ with input labels, gets the output label $[O]$, and extracts the output $O = f(I_g, I_e)$. The garbler can additionally verify if $[O]$ corresponds to $O$ to guarantee the integrity of the evaluation done by the evaluator. Extending from the simple garbling scheme, authenticated garbling~\cite{wang2017authenticated} enables maliciously secure two-party computation with additional authentication information, meaning both parties should follow the protocol honestly; otherwise, any behavior deviating from the protocol will be detected by the other party. We briefly discuss our protocol below, and the complete protocol is in Fig.~\ref{fig:protocol}. 

In our protocol, $V$ and $A$ act as the garbler and evaluator, respectively. First, $V$ and $A$ need to agree on an exchange time $t_e$ and a strategy profile $\boldsymbol{\pi}$ including $q$ and $\overline{p}$. The victim's true valuation is based on $t_e$ where $\theta^V = \psi(t_e)$. %
$V$ generates a garbled circuit $\mathcal{C}$ according to $\mathcal{M}$ and $\boldsymbol{\pi}$, and sends $\mathcal{C}$ to $A$. Next, $V$ and $A$ each generate two $k$-bit uniform random strings $(S^{V}_0,S^V_1)$ and $(S^{A}_0, S^A_1)$. $V$ directly encodes its inputs, then sends $[S^V_0]$,  $[S^V_1]$, and $[\hat{\theta}^V]$ to $A$. $A$ obtains $[S^A_0]$, $[S^A_1]$, and $[\hat{\theta}^A]$ via OT from $V$. %
Subsequently, $A$ evaluates $\mathcal{C}$ and extract the final ransom $r_f$. Finally, $A$ and $V$ respectively verify the correctness of $\mathcal{C}$ and $[r_f]$, ensuring neither party tampered with the computation before agreeing on $r_f$. 

Due to the probabilistic nature of $a_2^V$ and $a_4^V$, the input to $\mathcal{C}$ in our protocol must include randomness. To execute the probabilistic choices privately and fairly, we determine the branch by comparing two $k$-bit uniformly distributed strings $S_0$ and $S_1$ with thresholds corresponding to probabilities $\overline{p}$ and $q$. Here $S_0=S_0^V\oplus S_0^A$ and $S_1=S_1^V\oplus S_1^A$. Taking $S_0$ as an example, since $S_0^V$ and $S_0^A$ are private information provided by $V$ and $A$, respectively, neither of them would know the value of $S_0$. Note that $V$ or $A$ may attempt to choose $S_0^V$ or $S_0^A$ adversarially rather than uniformly, hoping to gain an advantage in the negotiation. However, this random number generation scenario resembles the classic matching pennies game, where the unique Nash equilibrium is to choose a uniform random value for each player.  %
This method ensures that the probabilistic choices within the circuit remain uninfluenced by either party, thereby guaranteeing the fairness and privacy for both $V$ and $A$.

Once the result is checked by both parties, $A$ and $V$ reach consensus on the ransom value $r_f$ generated by the circuit. The ransom-key exchange can then be executed either through a trusted third party, such as an insurance company, or via a smart contract on a blockchain, as proposed in \cite{vakilinia2021mechanism}. Specifically, the smart contract releases the ransom only after the victim deposits both the ransom $r_f$ and earnest money, the attacker submits the decryption key to the contract, and the victim verifies its correctness.
Upon confirmation, the ransom is transferred to the attacker, and the earnest money is refunded to the victim. This process can be regarded as a form of reputation system recording. Moreover, insurance companies can also serve as a reputation system, as cyber insurers often track behavior of ransomware groups, including negotiation frequency, duration, and reliability of decryption after payment \cite{stone2020fbi}.

As $\psi(t)$ is non-increasing, our mechanism gives the attacker financial incentives to quickly return the key. Therefore, our mechanism safeguards the privacy of both parties, incentivizes attackers to cooperate quickly with victims, and ensures a fair resolution for both sides.

\begin{figure*}[!t]
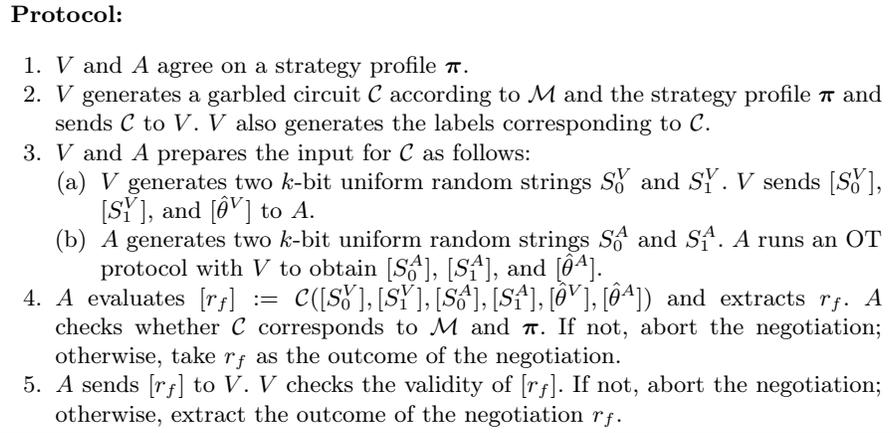

    \centering
    \fbox{
    \begin{minipage}{0.95\textwidth}
        \textbf{Protocol:}
        \begin{enumerate}
        \item $V$ and $A$ agree on a strategy profile $\boldsymbol{\pi}$.
        \item $V$ generates a garbled circuit $\mathcal{C}$ according to $\mathcal{M}$ and the strategy profile $\boldsymbol{\pi}$ and sends $\mathcal{C}$ to $V$. $V$ also generates the labels corresponding to $\mathcal{C}$.
        \item $V$ and $A$ prepares the input for $\mathcal{C}$ as follows:
        \begin{enumerate}
            \item[(a)] $V$ generates two $k$-bit uniform random strings $S^V_0$ and $S^V_1$. $V$ sends $[S^V_0]$,  $[S^V_1]$, and $[\hat{\theta}^V]$ to $A$.
            \item[(b)] $A$ generates two $k$-bit uniform random strings $S^A_0$ and $S^A_1$. $A$ runs an OT protocol with $V$ to obtain $[S^A_0]$, $[S^A_1]$, and $[\hat{\theta}^A]$.
        \end{enumerate}
        \item $A$ evaluates $[r_f] :=\mathcal{C}([S^V_0], [S^V_1], [S^A_0], [S^A_1], [\hat{\theta}^V], [\hat{\theta}^A])$ and extracts $r_f$. $A$ checks whether $\mathcal{C}$ corresponds to $\mathcal{M}$ and $\boldsymbol{\pi}$. If not, abort the negotiation; otherwise, take $r_f$ as the outcome of the negotiation.
        \item $A$ sends $[r_f]$ to $V$. $V$ checks the validity of $[r_f]$. If not, abort the negotiation; otherwise, extract the outcome of the negotiation $r_f$.
    \end{enumerate}
    \end{minipage}
    }
    \caption{A Privacy-Preserving Two-Party Negotiation Protocol}
    \label{fig:protocol}
\end{figure*}

\section{Garbled Circuit Implementation} \label{Sec:Imp}

We implemented our secure two-party protocol on the TinyGarble2 framework to instantiate our proposed mechanism \footnote[1]{The implementation is accessible to the public in the GitHub repository \url{https://GitHub.com/NomadShen/TinyGarble2.0}.}. TinyGarble2, utilizing EMP-tool~\cite{wang2016emp} as its backend, is an efficient \texttt{C++} framework for garbled-circuit-based S2PC in both semi-honest and malicious adversary models~\cite{hussain2020tinygarble2}.
The garbled-circuit implementation includes both parties inputting their respective $\hat{\theta}^V$ and $\hat{\theta}^A$, processing probabilistic choices, and finally reaching consensus on the ransom $r_f$. 

\subsection{Implementation}
Our implementation assumes V and A have already agreed on a strategy profile $\boldsymbol{\pi}$. The circuit has $3$ inputs from $V$, which are $S^V_0$, $S^V_1$, and $\hat{\theta}^V$. Correspondingly, the values that $A$ inputs into the circuit include $S^A_0$, $S^A_1$, and $\hat{\theta}^A$. Consistent with Step 4 of the Protocol in Fig.~\ref{fig:protocol}, the circuit's output is the garbled form of the ransom value $[r_f]$, but the plaintext value $r_f$ is extractable by both parties from $[r_f]$ to ensure the result is fairly distributed to both parties.

For probabilistic choice, we precompute $\overline{p}_{scale}=\lfloor \overline{p}\cdot 2^k\rfloor$, and then complete the probabilistic choice simply through comparison between $\overline{p}_{scale}$ and $S_0$. Additionally, for the rejection in $\mathcal{M}$, the circuit assigns $r_f$ a value of $0$. 

The multiplication and division operations involved in $\mathcal{M}$ include $q\hat{\theta}^V$ and $r_2/q$, where $q\in(0,1/2]$. In practice, $q$ is typically not very close to zero, so we precompute the values of $q$ and $\frac{1}{q}$ and perform scaling as $\overline{p}_{scale}$. Similarly, we compute $q\cdot \hat{\theta}^V=q_{scale}\cdot \hat{\theta}^V/{2^k}$ and $r_2/q = r_2\cdot (\frac{1}{q})_{scale}/{2^k}$. Here, division by $2^k$ is implemented as a right-shift operation. We set the bitwidth of the input $\hat{\theta}$ to $k_{\theta}$ and the bitwidth of the randomness to $k$. %

\subsection{Evaluation}

We evaluate the execution time of our protocol implemented on the TinyGarble2 framework. We run $V$ and $A$'s programs on the same computer using two separate threads, with the two parties communicating by the TCP protocol. Evaluation is performed on an Intel Core i7-1250U CPU with 8GB RAM.

In the experiment, we evaluate the protocol with different bitwidths of $k_{\theta}$ and $k$, which affect the precision of probabilities and multiplications.
Table~\ref{tab:evaluation} shows the execution time under different input bitwidths. The results show that execution time increases with input bitwidth, but all computations complete within $53$ ms, which is well within practical online negotiation requirements compared to the human negotiation process. %

\begin{table}[t]
    \centering
    \caption{Execution Time Evaluation Results}
    \label{tab:evaluation}
    \setlength{\tabcolsep}{20pt}
    \begin{tabular}{ccc}
        \hline
        \textbf{$k_\theta$} & \textbf{$k$} & \textbf{Execution time} \\
        \hline
        $8$ & $8$ & $18.96$ ms \\
        $8$ & $16$ & $32.56$ ms \\
        $8$ & $32$ & $41.02$ ms \\
        $16$ & $8$ & $34.95$ ms \\
        $16$ & $16$ & $40.25$ ms \\
        $16$ & $32$ & $52.18$ ms \\
        \hline
    \end{tabular}
\end{table}

\section{Conclusion} \label{Sec:Con}

This paper presents a game-theoretic and privacy-preserving approach to ransomware negotiation, a critical yet underexplored phase of ransomware incidents. We first highlight the importance of negotiation and reputation systems to both attackers and victims through multistage game-theoretical analysis. Then, we model the ransomware negotiation process as a finite-horizon alternating-offers bargaining game. Our analysis captures how strategic behavior influences outcomes in complete and incomplete information settings. To operationalize our findings, we design a Bayesian incentive-compatible negotiation mechanism that allows both parties to reach an agreement quickly while preserving privacy. Our implementation using garbled circuits enables efficient and secure negotiation without revealing sensitive data. To the best of our knowledge, this is the first work to integrate a formal bargaining model with a privacy-preserving, automated negotiation mechanism tailored to ransomware. By combining theoretical insights with practical cryptographic implementation, our work provides both a deeper understanding of ransomware negotiation dynamics and a viable pathway for more secure and efficient post-infection response strategies.

\end{document}